\documentclass[twoside]{IEEEtran}

\usepackage{amsmath}
\usepackage{latexsym}
\usepackage{amssymb}

\usepackage{upref}
\usepackage{theorem}
\usepackage{graphicx}
\usepackage{psfrag}
\usepackage{algorithmic}
\usepackage{algorithm}
\usepackage{verbatim}
\usepackage{cite}
\usepackage{color}
\usepackage{mathrsfs}
\usepackage{url}

\newtheorem{theorem}{Theorem}
\newtheorem{lemma}[theorem]{Lemma}
\newtheorem{claim}[theorem]{Claim}

\newcounter{enumrom}
\renewcommand{\theenumrom}{(\roman{enumrom})}

\makeatletter
\renewcommand{\@endtheorem}{\endtrivlist}
\makeatother

\makeatletter
\renewcommand{\thefigure}{{\@arabic\c@figure}}
\renewcommand{\fnum@figure}{{\bf Figure\,\thefigure}}
\makeatother

\newcommand{\ml}[1]{#1}
\newcommand{\sj}[1]{#1}
\newcommand{\bikd}[1]{#1}

\newcommand{\suppress}[1]{}

\newcommand{\bp}{\bar{p}}

\newcommand{\mbf}{\mathbf}

\newcommand{\mess}{u}
\newcommand{\ind}{r}

\newcommand{\Ca}{{C}}

\def\cX{\mbox{$\cal{X}$}}
\def\cY{\mbox{$\mc{Y}$}}

\def\cI{{\mbox{${\cal{I}}$}}}

\def\cV{\mbox{$\cal{V}$}}
\def\cE{\mbox{$\cal{E}$}}
\def\cU{\mbox{$\cal{U}$}}

\def\rR{\mathfrak{R}}
\def\sR{\mathscr{R}}

\def\e{\varepsilon}

\def\bx{{\bf x}}
\def\be{{\bf e}}
\def\by{{\bf y}}

\def\bU{{\bf U}}

\def\be{{\bf e}}
\def\bl{{n}}
\def\rate{{R}}
\newcommand{\Graph}{{{\cal G}}}

\def\mc{\mathcal}

\def\p{{\rho}}

\def\01{\{0,1\}}
\def\bsc{\mathsf{BSC}}

\newcommand{\remove}[1]{}

\begin{document}

\IEEEoverridecommandlockouts

\title{Upper Bounds on the Capacity of Binary Channels with Causal Adversaries}

\author{Bikash Kumar Dey,~\IEEEmembership{Member,~IEEE},
Sidharth Jaggi,~\IEEEmembership{Member,~IEEE},
Michael Langberg,~\IEEEmembership{Member,~IEEE},
Anand~D.~Sarwate,~\IEEEmembership{Member,~IEEE}
\thanks{Manuscript received April 11, 2012; revised September 26, 2012.  Date of current version \today.
Authors are in alphabetical order. Work supported in part by ISF grant 480/08, RGC GRF grants 412608 and 412809, RGC AoE grant on Institute of Network Coding, established under the University Grant Committee of Hong Kong, CUHK MoE-Microsoft Key Laboratory of Humancentric Computing and Interface Technologies, the Bharti Centre for Communication in IIT Bombay, India, and the California Institute for Telecommunications and Information Technology (CALIT2) at UC San Diego.
A preliminary version of this work was presented at the IEEE International Symposium on Information Theory, Cambridge, MA, USA, July 2012~\cite{isit2012}.
}
\thanks{B.K. Dey is with the Department of Electrical Engineering
at the Indian Institute of Technology Bombay, Powai, Mumbai 400 076, India.  
Email : \texttt{bikash@ee.iitb.ac.in}.  
S. Jaggi is with the Department of Information Engineering
at the Chinese University of Hong Kong, Shatin, N.T., Hong Kong.  
Email : \texttt{jaggi@ie.cuhk.edu.hk}.
M. Langberg is with the Department of Mathematics and Computer Science
at The Open University of Israel, 108 Ravutski St., Raanana 43107, Israel.
Email : \texttt{mikel@openu.ac.il}.
A.D. Sarwate is with the Toyota Technological Institute at Chicago,
6045 S. Kenwood Ave., Chicago, IL 60637, USA. 
Email : \texttt{asarwate@ttic.edu}.
}
\thanks{Communicated by T. Weissman, Associate Editor for Shannon Theory.}
\thanks{Digital Object Identifier 10.1109/TIT.2013.XXXXXXXX}
\thanks{Copyright (c) 2012 IEEE. Personal use of this material is permitted. �However, permission to use this material for any other purposes must be obtained from the IEEE by sending a request to \texttt{pubs-permissions@ieee.org}.}
}

\maketitle
\begin{abstract}
In this work we consider the communication of information in the presence of a
{\em causal} adversarial jammer.
In the setting under study, a sender wishes to communicate a message to a receiver by transmitting a codeword $\bx=(x_1,\dots,x_\bl)$ bit-by-bit over a communication channel.  The sender and the receiver do not share common randomness. The adversarial jammer can view the transmitted bits $x_i$ one at a time, and can change up to a $p$-fraction of them.
However, the decisions of the jammer must be made in a {\em causal} manner.
Namely, for each bit $x_i$ the jammer's decision on whether to corrupt it or not 
must depend only on $x_j$ for $j \leq  i$. This is in contrast to the ``classical'' adversarial jamming situations in which the jammer has no knowledge of $\bx$, or knows $\bx$ completely.
In this work, we present upper bounds (that hold under both the average and maximal probability of error criteria) on the capacity which hold for both deterministic and stochastic encoding schemes.
\end{abstract}

\begin{IEEEkeywords}
  channel coding, arbitrarily varying channels, jamming
\end{IEEEkeywords}



\section{Introduction}

Alice wishes to transmit a message $\mess$ to Bob over a binary-input
binary-output channel. To do so, she encodes $\mess$ into a length-$\bl$ binary
vector $\bx$ and transmits it over the channel. However, the channel is
controlled by a malicious adversary Calvin who may observe the transmissions,
and attempts to jam communication by flipping up to a \bikd{$p$ fraction} of the bits \bikd{transmitted by Alice.}
Since he must act in a causal manner, Calvin's decisions on
whether or not to flip \sj{the} bit $x_i$ must be a function solely of the bits
$x_1,\ldots,x_i$ he has observed thus far. This communication scenario models
jamming by an adversary who is limited in his jamming capability (perhaps due to
limited transmit energy) and is causal.
This {\em causality} assumption is reasonable for many communication channels,
both wired and wireless.
Calvin can only corrupt a bit when it is transmitted (and thus its error is
based on its view so far).
To decode the transmitted message, Bob waits until all the bits have arrived.

In this paper we investigate the information-theoretic limits of communication
in this setting.  We stress that in our model Calvin knows everything that both Alice and Bob do -- there is no {\it shared secret or common randomness} \sj(a model where such a shared secret {\it may} be allowed has been considered in the literature pertaining to {\it Arbitrarily Varying Channels}, discussed further in \bikd{Section}~\ref{subsec:prev}).  However, we make no assumptions about the computational tractability of Alice, Bob, or Calvin's encoding, decoding and jamming processes.  
Our main contribution in this work is a {\it converse} that helps \sj{to}
make progress towards a better understanding of the communication \sj{{\it rates } (average number of bits per channel use)} 
achievable against a causal adversary. Specifically, we describe and analyze a
novel jamming strategy for Calvin and show that it (upper) bounds the rate 
of communication regardless of the coding strategy used by Alice and Bob.
This jamming strategy results in Calvin being able to force Bob's {\it average} probability of decoding error over all of Alice's messages to be bounded away from zero (and hence correspondingly also his {\it maximum} probability of error).

\subsection{Previous and related work}\label{subsec:prev}

Many of the following works deal with related channels; we restrict
our discussion mostly to the binary-input binary-output case, \sj{except where specifically indicated otherwise.}

\noindent{\bf Coding theory model:}  A very strong class of adversarial channels
is one where Calvin is {\it omniscient} -- he knows Alice's entire
codeword $\bx$ prior to transmission and can tailor the pattern of up to $p \bl$
bit-flips to each specific transmission.  This is the {`` worst-case noise''} model
studied in coding theory. \sj{In this model there is no randomness in code design, and it is desired that Bob {\it always} decodes correctly.}  For binary channels, characterizing the capacity
 has been an open problem for several decades.  
The best known upper bound is due to McEliece {\it et al.}~\cite{McEliece77} 
as the solution of an
LP, and the best known achievable scheme corresponds to codes suggested by
Gilbert and Varshamov~\cite{Gil52,Var57}, which achieve a rate of $1-H(2p)$.
Improving either of these bounds would be a significant
breakthrough.\footnote{As is often the case, results for channels over ``large"
alphabets are significantly easier. \sj{In the ``intermediate'' alphabet-size regime, wherein the alphabet is of size at least $49$, advances in Algebraic-Geometry codes over the last three decades (see~\cite{TsaVN:07} for a survey) have resulted in codes exceeding the Gilbert-Varshamov bound.} For alphabets larger than $\bl$, the bound
of $1-2p$ due to Singleton~\cite{Singleton:64} is known to be achievable in a
	computationally efficient manner via Reed-Solomon codes~\cite{ReedSolomon:60}. }

\noindent{\bf Information theory model:}  A much weaker class of adversarial
channels is one where Calvin generates bit flips in an i.i.d. manner with probability $p$ \sj{and Bob must decode correctly ``with high probability'' over the randomness in Calvin's bit-flips.} 
The original work of Shannon~\cite{Shannon48} effectively characterized 
the capacity of this {\em binary symmetric channel} $\bsc(p)$.  The capacity $1-H(p)$ in Shannon's setting (for crossover probability $p$) is strictly greater than that of the \sj{coding theory} model.

\noindent{\bf Causal adversarial model:} The class of channels considered in
this work, {\it i.e.}, that of {\it causal adversaries}, fall\sj{s} in between \sj{the above}
two extremes. In one direction this is because a causal adversary is certainly
no stronger than an omniscient adversary, since he cannot tailor his jamming
strategy to take into account Alice's future transmissions. Indeed, the work
of \sj{Haviv and Langberg}~\cite{HL11} indicates that (for $2p< H^{-1}(1/2)\simeq 0.11$) rates strictly better than those
achievable by \sj{Gilbert-Varshamov} codes \sj{~\cite{Gil52,Var57}} against an omniscient adversary are achievable against a
causal adversary. However, since it is still unknown whether Gilbert-Varshamov codes are
optimal against omniscient adversaries, it is unknown whether causal adversaries
are indeed strictly weaker than omniscient adversaries. \sj{Nonetheless, the Gilbert-Varshamov bound and the bound of~\cite{HL11} indicate that for $p < 1/4$ the capacity under causal adversaries is bounded away from zero.}

In the other direction, the causal adversarial model under study is at least as strong as the information theoretic model in which Calvin generates bit flips in an i.i.d. manner.
Specifically, \sj{if $p\leq 1/2$, for any $\delta_p > 0$ and sufficiently long block-length $\bl$ a causal adversary can ignore the transmitted codeword seen so far and just mimic the behavior of a binary symmetric channel $\bsc(p-\delta_p)$ -- with high probability he does not exceed his budget of $p\bl$ bit-flips. Similarly, if $p > 1/2$, Calvin simply mimics the behavior of a $\bsc(1/2)$.}
This implies that when communicating in the presence of causal adversaries with jamming capabilities that are parametrized by $p$, $1-H(p)$ is an upper bound on the achievable rate for $p\leq 1/2$, and no positive rate is achievable for $p>1/2$.
Improving over this na\"ive upper bound \sj{(and hence narrowing the gap to the \bikd{lower} bound of~\cite{HL11})} is the focus of the paper.

The improved upper bounds we present hold for general coding schemes that allow Alice to encode a message $\mess$ to one of several possible codewords $\bx
\in \{\bx(\mess,\ind)\}$, where $\ind$ is a random source available to Alice but unknown to either Bob or Calvin.
Such general coding schemes are referred to as {\em stochastic} coding schemes.
We stress that in such schemes there is no shared randomness between Alice and Bob, and the source of randomness in Alice's encoder is solely known to Alice.

\noindent{\bf Arbitrarily Varying Channels:}  Our model is a variant of the arbitrarily varying channel (AVC) model \cite{BlackwellBT:60random}.  The AVC model where the adversary has access to the entire codeword was considered by Ahlswede and Wolfowitz \cite{AhlswedeW70:binary,Ahlswede:73fback} but received little attention since \cite[Problem 2.6.21]{CsiszarKorner}.  General AVC models have been extended to include channels with constraints  on the adversary (such as $p\bl$ bit flips) for cases where the adversary has no access to
the codeword~\cite{CsiszarN:88positivity}, or has access to the full codeword~\cite{SarwateG:10csi}.  For binary channels in which the jammer has knowledge of the entire codeword $\bx$, \cite{Langberg:04focs} showed that $O(\log n)$ bits of common randomness is sufficient to achieve the optimal rate of $1 - H(p)$ (\sj{and the work in~\cite{Smith:07scrambling}
investigated computationally efficient constructions of such codes}).  However, issues of causality have only been studied in the context of randomized coding (when the encoder and decoder share common randomness), but not for deterministic codes or stochastic encoding.

\noindent{\bf Delayed adversaries:}  The {\it delayed adversary} model
was studied in~\cite{Langberg:08} and~\cite{DJLS10}.  In this model, the
jammer's decision on whether to corrupt $x_i$ must depend only on $x_j$ for \sj{$j
\leq  i - D\bl$} for a delay parameter \sj{$D \in [0,1]$}.  The case of \sj{$D=0$} is
exactly the causal setting \sj{studied in this work}, and that of \sj{$D = 1$} corresponds to the
``oblivious adversary" studied \sj{by Langberg}~\cite{Langberg:08}.  
\sj{In this oblivious adversary setting the work of~\cite{GuruswamiS:09stoc} demonstrates computationally efficient code constructions that achieve information-theoretically rate-optimal throughput of $1-H(p)$ for all $p<1/2$.}

\sj{In a different line of work, Dey et al.~\cite{DJLS10}}  
showed that for a large class of channels, the capacity for delay \sj{$D>0$} equals 
that of the constrained AVC model~\cite{CsiszarN:88constraints}.  In particular, a positive
delay implies that the optimal rate $1 - H(p)$ is achievable against a delayed
adversary \sj{over a binary-input binary-output channel}.  
In this paper we show that a causal adversary is strictly
stronger than a delayed adversary \sj{with $D>0$} for all $p > 0.0804$. \sj{For $p$ smaller than this value our techniques do not help separate the capacity regions of these two models.}

\sj{\noindent{\bf Causal and delayed adversaries for ``large alphabets'':} In the {\it large alphabet} setting (where the alphabet-size is allowed to grow without bound with increasing block-length), Dey et al.~\cite{DeyJL:09allerton} give a full characterization of the capacity-region of several variants of both the causal adversary and the delayed adversary models. They further give computationally efficient codes achieving every point in the capacity regions for the models considered.
In general, in the large alphabet regime code design is easier than in the binary alphabet regime (that is the primary focus of this work) since with large alphabets, a ``few random hashes'' can be hidden inside each symbol with asymptotically negligible rate-loss. These hashes aid the decoder in detecting the adversarial attack pattern and correcting for it. In the binary alphabet setting this technique is not applicable -- this is one of the bottlenecks in further narrowing the gaps between outer and inner bounds for the model considered in this work.}

\noindent{\bf Previous attacks:} 
\sj{This work} continues our preliminary work on binary causal channels~\cite{LangbergJD09} \sj{(and a related result of Guruswami and Smith~\cite{GuruswamiS:09stoc})}, which proposed an upper bound using the so called ``wait-and-push" attack.  This work improves on this earlier work in two aspects -- specifically the bound presented is tighter, and holds also for stochastic \sj{encoding}. 
\ml{\footnote{
\ml{
For completeness, we specify the two major differences between this paper and~\cite{LangbergJD09}.  First, we propose a different two-phase attack (``babble-and-push'') which gives a tighter outer bound than the previous attack (``wait-and-push'').  In ``wait-and-push,'' Calvin passively eavesdrops in the first phase uses this information to design an error vector to confuse Bob in the second phase.  In our new attack, Calvin instead injects noise in the first phase to increase Bob's uncertainty about Alice's transmissions.  However, we must carefully choose the number of bit-flips Calvin injects in this ``babble'' phase to obtain a tighter outer bound, because Calvin must trade-off between using bit-flips to increase Bob's uncertainty and to push to an alternative codeword in the second phase.
The second improvement in this paper is that we prove that the ``babble-and-push'' attack works even when Alice and Bob use stochastic encoding ({\it i.e.}, for each message $\mess$ she has, Alice may choose to transmit one of multiple possible codewords $\bx(\mess)$, with an arbitrary random distribution over the set of codewords).  Our bounds therefore hold for general codes, as opposed to previous work~\cite{LangbergJD09}, where the outer bound was proved for codes in which each message $\mess$ corresponded to a unique $\bx(\mess)$ {\it deterministically} chosen by Alice.} 
}}


\subsection{Main result}

Our improved bounds are given in the following theorem, and are depicted (in
comparison with the previous bounds) in Figure~\ref{fig:bounds}.
For any $\bp \in  [0,p]$, let 
	$
	\alpha(p,\bp) = 1-4(p-\bp).
	$
In what follows, $C(p)$ is the capacity of the causal channel under study.
For precise definitions and model see Section~\ref{sec:model}.

\begin{theorem}
\label{the:det}
For $p \in [0,1/4]$, the capacity $C(p)$ of a binary causal adversary channel with
constraint $p$ satisfies:
\[
\Ca(p) \leq \min_{\bp \in [0,p]} \left [
\alpha(p,\bp)\left(1-H\left(\frac{\bp}{\alpha(p,\bp)}\right)\right) \right ].
\]
For $p > 1/4$ the capacity $C(p) = 0$.
\end{theorem}
\ml{A few remarks are in order.
Notice that in the regime $\bp \leq p \leq 1/4$ it holds that $\bp \leq \alpha(p,\bp)$ and thus $\frac{\bp}{\alpha(p,\bp)}$ in the expression of Theorem~\ref{the:det} is at most of value $1$.
We show in Appendix~\ref{app:min} that 
the optimum $\bar{p}$ in the computation of $\Ca(p)$ is
	\begin{align*}
\min\left\{p, \frac{3(1-4p)}{\left (2+\sqrt[3]{1592+24\sqrt{33}} + \sqrt[3]{1592-24\sqrt{33}} \right )}\right\} & \\
&\hspace{-2in} \simeq \min\left\{p, \frac{1-4p}{8.4445}\right\}.
	\end{align*}
}
Namely, for \sj{$p$ greater than approximately $0.0804$}, the capacity $C(p)$ is bounded away from $1-H(p)$ and for $p$ \sj{less than this value} our bound equals $1-H(p)$ (in the latter case we get $\bar{p}=p$).
For $p = 1/4$ the new strategy we propose for Calvin shows that no positive rate is achievable; when $p > 1/4$ Calvin can simply mimic the case $p = 1/4$.

\begin{figure}
\centering
\includegraphics[width=3.4in]{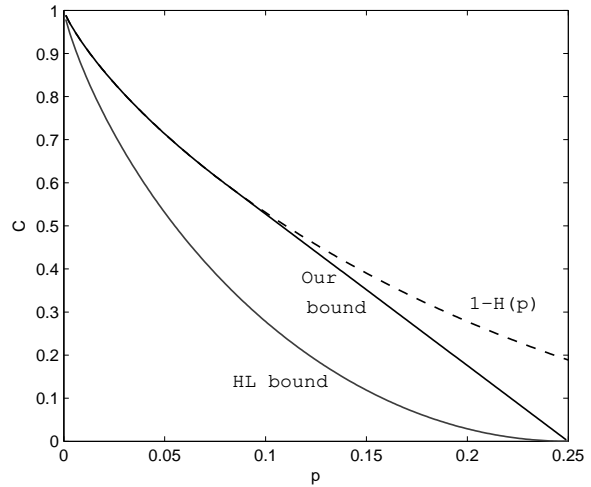}
\caption{We plot previous bounds related to the channel at hand compared to our bound. \bikd{The {\em upper} bound of $1-H(p)$ corresponds to the binary symmetric channel. The {\em lower} bound~\cite{HL11} (denoted HL) is based on an evaluation of the parameters specified by Haviv and Langberg~\cite{HL11} and it slightly improves on the Gilbert-Varshamov bound $1-H(2p)$.} 
Our improved bound appears in between.}
\label{fig:bounds}
\end{figure}

\subsection{Techniques and Proof Overview}
\label{sec:tech}

\sj{To prove Theorem~\ref{the:det} we show that no matter which
encoding/decoding scheme is used by Alice and Bob, there exists a strategy for Calvin 
that does not allow communication at rate higher than $\Ca(p)$ .
Specifically, we demonstrate that whenever Alice and Bob attempt to communicate
at a rate higher than $\Ca(p)$,
there exists a causal jamming strategy (that in general depends on Alice and Bob's encoding/decoding strategy) that allows Calvin to enforce a constant
probability of error bounded away from zero.}
{\sj{More precisely, for any block-length $\bl$, any $\e>0$ and any encoding/decoding scheme of Alice and Bob of rate $\left(C(p)+ \sqrt{\frac{c}{n}}\right) +\e$, Calvin can cause a decoding error \sj{probability}
of at least $\e^{O(1/\e)}$.}

\sj{At a high level,} Calvin uses a two-phase {\em ``babble-and-push''} strategy.
In the first phase of $\ell(p)$ channel uses, Calvin ``babbles'' by behaving
like a $\bsc(\bar{p})$ for some $\bar{p}$ chosen as a function of $p$.
In the second phase of $n - \ell(p)$ channel uses, Calvin randomly selects a codeword
from Alice and Bob's codebook \sj{that} is consistent with what Bob has 
received so far.  Calvin then ``\sj{randomly} pushes'' the remaining
part of Alice's codeword towards his selected codeword \sj{({\it i.e.}, in every location in the ``push'' phase where Alice's codeword bit differs from his selected codeword, he adds a bit-flip with probability half)}.
A decoding error occurs if Calvin is able \bikd{to} push the transmitted codeword half the distance towards the codeword selected by Calvin (via a standard symmetrization argument \cite{CsiszarN:88positivity}).

Roughly speaking, the first phase allows Calvin to gain information regarding which codeword was transmitted by Alice, while the second phase allows Calvin to use this information in order to design a corresponding symmetrization-based jamming strategy. 

\sj{In Section~\ref{sec:main} we present the proof of our main result, that of the outer bound on the capacity of online adversaries. Section~\ref{sec:det} then improves on this result (by giving a tighter bound on the probability of error) for the special case of deterministic encoders (rather than the general stochastic encoders considered in Section~\ref{sec:main}).}


\section{Model \sj{and preliminaries}}
\label{sec:model}

\sj{We first reprise some standard notation. Let $d_H(\cdot,\cdot)$  denote the {\it Hamming distance} function between two vectors (number of locations in which two vectors of the same length differ). The {\it Hamming weight $wt_H(\mbf{x})$} of a vector $\mbf{x}$ is the Hamming distance between that vector and the all-zeros vector.
Let $\log(.)$ denote the binary logarithm, here and throughout. As is common, the notation $H(A)$ is used to denote the (binary) {\it entropy} of a random variable $A$, $H(A|B)$ to denote the {\it conditional entropy of $A$ given $B$}, and $I(A;B)$ to denote the {\it mutual information between $A$ and $B$.} Also, for any real number number $x \in (0,1)$, $H(x)$ denotes the {\it binary entropy function}.
Properties of and inequalities between these functions are referenced at the point in the text where needed. The {\it indicator function} $\mbf{1}(condition)$ takes value $1$ if $condition$ is true, and $0$ otherwise.
}

Let the input and the output alphabets of the channel be $\cX$ and $\cY$
respectively.
For any positive integer $k$, let $[k] = \{1,2,\ldots,k\}$. 
We let $\cU = [2^{nR}]$ denote Alice's message set, and $\bU$ denote
the message random variable uniformly distributed in $\cU$.
A \emph{deterministic code} of rate $R$ and block-length $n$ is a pair of maps
$\mathcal{C}_d = (\Phi,\Psi)$ where $\Phi : \cU \to \cX^n$ and $\Psi :
\mc{Y}^n \to \cU$ are deterministic maps.  The map $\Phi$ is called the
encoder and the map $\Psi$ is called the decoder.  

A \emph{code with stochastic encoding and decoding} of rate $R$ and block-length $n$ is a pair
of maps $\mathcal{C}_s = (\Phi,\Psi)$ where $\Phi : \cU \to \cX^n$ and $\Psi : \mc{Y}^n \to \cU$ are probabilistic maps.
The random map $\Phi$ gives a probability distribution $\p(\cdot|\mess)$
on $\cX^n$ for every $\mess \in \cU$. 
The mapping $\Psi(\by)$ is a random variable taking values from $\cU$. 
The encoding $\Phi$ is equivalently 
represented by first picking a random variable $\rR$ from a set
$\sR$ according to a conditional distribution 
$\p_{\rR|\bU}(.|u)$, and then \sj{applying} a deterministic
encoder map $\Phi: \cU \times \sR \to \cX^n$.
Note that our definition
does not preclude there existing pairs $(\mess,\ind)$ and $(\mess',\ind')$
such that $\Phi(\mess,\ind) = \Phi(\mess',\ind')$.
As we are addressing upper bounds on the capacity $C(p)$ in this work, it is crucial to prove our results in the stochastic setting above \sj{-- any bounds proved in the stochastic setting also hold in the deterministic setting.}

A \emph{causal adversarial strategy} of block-length $n$ is a sequence of \sj{
(possibly random) mappings $\mathsf{Adv} = \{ f_{\mc{C}}^{(i)} : i \in [n] \}$. Here each $f_{\mc{C}}^{(i)}
: \cX^i  \times \mc{E}^{i-1}  \to \mc{E}$ depends on 
$\mc{C}$, and for each time $i \in [n]$ 
chooses an {\it action at time $i$,} $e_i = f_{\mc{C}}^{(i)}(x_1, \ldots,
x_i) \in \mc{E}$ -- the inputs to $f_{\mc{C}}^{(i)}$ are
the past and current channel inputs
$(x_1,x_2,\ldots,x_i)$ and its own previous actions $(e_1,e_2,\ldots,e_{i-1})$. 
The resulting {\it channel output at time $i$} is $y_i = x_i + e_i$.  }
In our setting $\mc{E} = \{0,1\}$.
The strategy obeys constraint $p$ if the Hamming weight $\|\mbf{e}\|=\sum_{i=1}^{n} e_i$ of \sj{$\mbf{e}=(e_1,\dots,e_n)$}  is at most $p n$ over
the randomness in the message, encoder, and strategy.  For a given adversarial
strategy and an input codeword $\mbf{x}$, the strategy produces a (possibly
random) $\mbf{e}$ and the output is $\mbf{y} = \mbf{x} \oplus \mbf{e}$.  Let
$\Pr_{\mathsf{Adv}}(\mbf{y} | \mbf{x})$ denote the probability of an
output $\mbf{y}$ given an input $\mbf{x}$ under the strategy $\mathsf{Adv}$ \sj{where this strategy might depend on $(u,r)$ via the adversary's causal observations of $\bx$ -- to simplify notation we henceforth do not make this explicit}.
When the block-length is understood from the context, let $\mathsf{Adv}(p)$ 
denote all adversarial strategies obeying constraint $p$.

The (average) probability of error for a code with stochastic encoding and decoding is given by
	\begin{align}
        \bar{\varepsilon} &= \max_{\mathsf{Adv} \in \mathsf{Adv}(p)}
                \frac{1}{2^{R n}}
                \sum_{\mess=1}^{ 2^{Rn} }
                \sum_{\ind \in \sR} \p_{\rR|\bU}(\ind | \mess) 
                \nonumber \\
                &\hspace{1.1in} 
                \sum_{\mbf{y}} \Pr_{\mathsf{Adv}}(\mbf{y} | \Phi(\mess,\ind) )
                        \Pr\left( \Psi(\mbf{y}) \ne \mess \right),
        \label{eq:stoc:err}
        \end{align}
where the probability $\Pr\left( \Psi(\mbf{y}) \ne \mess \right)$ is over
any randomness in the decoder \sj{(but there is {\it no} conditioning on $\ind$ since shared randomness between the encoder and the decoder is not allowed)}.  We can interpret the errors as the error in
expectation over Alice choosing a message $\bU=\mess$ and a codeword $\mbf{x}
= \Phi(\mess,\ind)$ according to the conditional distribution $\p( \mbf{x} |
\mess)$.

A rate $R$ is \emph{achievable} against a causal adversary 
under average error if for every $\delta
> 0$ there exist infinitely many block-lengths $\{\bl_i\}$, such that for each $n_i$ there is an $n_i$ block-length (stochastic) code 
 of rate at least $R$ and 
average probability of error at most $\delta$. The supremum of all 
achievable rates is the capacity. We denote by $C(p)$ the capacity of the channel corresponding to adversaries parametrized by $p$.

Consider a code of block-length $n$, rate $R$ and
error probability $\delta$.
We can, without loss of generality (w.l.o.g.), assume that the encoding probabilities $\{
\p(\mbf{x} | \mess) : \mbf{x} \in \{0,1\}^n, \mess \in [2^{nR}] \}$ are
rational.  To see why this is the case, note that for any small $\eta > 0$ we
can find rational numbers $\{\tilde{\p}(\mbf{x} | \mess)\}$ such that
$\p(\mbf{x} | \mess) - \eta \le \tilde{\p}(\mbf{x} | \mess) \le \p(\mbf{x} |
\mess)$.  Now consider a code with encoding probabilities $q(\mbf{x} | \mess) =
\tilde{\p}(\mbf{x} | \mess)$ for $\mbf{x} \ne \mbf{0}$ and assign the remaining
probability to $\mbf{0}$.  Under the same decoder,
this code has error probability at most 
$\delta + 2^{n + nR} \eta$, but since $\eta$ was arbitrary, 
the error is at most $2 \delta$. 

Now, for a given stochastic code, let $N$ be the least common multiple
of the denominators of $\p (\bx|u)$ for all $\bx, u$. 
\sj{Each codeword $\bx$ of $u$ can be treated as} $N\p (\bx|u) $ copies of the same codeword
with conditional probability $1/N$ each. So we can equivalently
associate a random variable $\rR$ with $|\sR| = N$ s.t. 
the conditional distribution $\p_{\rR|\bU} (\cdot|u)$
is uniform, and the encoding map $\Phi (u, \cdot)$ is not necessarily \sj{injective}.
Since we consider \sj{the} uniform message distribution, henceforth, w.l.o.g., we assume
that the joint distribution $\p_{\bU,\rR}(\cdot, \cdot)$ is uniform.

\sj{
We use a version of Plotkin's bound~\cite{Plo:60} in our proof.  This result gives an upper bound on the number of codes in any binary code with a given minimum distance.
\begin{theorem}[Plotkin bound~\cite{Plo:60}] \label{thm:plo}
There are at most $\frac{2d_{min}}{2d_{min}-n}$ codewords in any binary code of block-length $n$ with minimum distance $d_{min} > n/2$.
\end{theorem}
}

\section{Proof of Theorem~\ref{the:det}  \label{sec:main}}

\sj{In this section we analyze an adversarial attack for the general case of stochastic encoders and decoders.  For fully deterministic codes the analysis is more combinatorial and the error bounds are somewhat better, as shown in Section~\ref{sec:det}.}

Let $p \in [0,1/4]$ and let $\bp \leq p$. \sj{Without loss of generality we assume that $pn$ is an integer -- if not, Calvin can simply choose the largest $p'$ smaller than $p$ such that $p'n$ is an integer. Asymptotically in $n$, the effect of this quantization on our outer bound is negligible.}
Let $\e >0$.
In what follows we prove that the rate of communication over the causal adversarial channel (with parameter $p$) is bounded by 
	\begin{align}
	\rate \leq C + \e, 
	\label{eq:assum_start}
	\end{align}
where
	\begin{align}
	C = \alpha(p,\bp)\left(1-H\left(\frac{\bp}{\alpha(p,\bp)}\right)\right) 
	\end{align}
and
	\begin{align}
	\alpha = \alpha(p,\bp) = 1-4(p-\bp),
	\label{eq:assum_end}
	\end{align}
\noindent as defined in Theorem~\ref{the:det}. 
Namely, if \eqref{eq:assum_start}--\eqref{eq:assum_end} is violated, for any sufficiently large block-length $\bl$, and any ($\bl$-block stochastic) code $\mathcal{C}_s=(\Phi,\Psi)$ shared by Alice and Bob, there exists an adversarial jammer $\mathsf{Adv}$ that can impose a constant decoding error.
The decoding error we obtain will depend on $\e>0$.
 
For $\bp = p$, the adversary can generate a noise sequence to simulate a $\bsc$ with crossover \bikd{probability} arbitrarily close to $p$, which yields an upper bound \bikd{of $1 - H(p)$ on the capacity}.  We therefore assume that $p-\bp>0$ and that $\e<2(p-\bp)$.  We show that for such $\e > 0$ there cannot exist a sequence of codes \sj{(each with rate at least $C + \e$) of increasing block-length $n$, such that the probability of error of these codes converges to $0$ asymptotically in $n$}.  To do so we
will consider block-lengths $\bl \geq \Omega(\e^{-2})$.  Note that this argument does not provide lower
bounds on the error of codes of a given block-length, but instead shows a bound on the capacity.
We elaborate on this point at the end of the proof.

Our converse bound is based on a particular two-phase adversarial strategy for Calvin that we call ``babble-and-push.''  
Let $\ell = (\alpha + \e/2) n$ and without loss of generality assume $\ell \in \mathbb{N}$.  
For a vector $\mbf{z}$ of length $n$, let $\mbf{z}_1 = (z_1, z_2, \ldots, z_{\ell})$ and $\mbf{z}_2 = (z_{\ell+1}, z_{\ell + 2}, \ldots, z_{n})$.
In what follows, $\mbf{z_1}$ will correspond to the first phase of Calvin's attack, while $\mbf{z}_2$ corresponds to the second phase.  \sj{For $p > 0$ the strategy is given as follows.}

\begin{itemize}
\item \textbf{(``Babble'')} Calvin chooses a random subset $\Gamma$ of $\bar{p} n$ indices uniformly from the set \sj{of all $(\bar{p} n)$-sized subsets} of $\{1, 2, \ldots, \ell\}$.  For $i \in \Gamma$, Calvin flips bit $x_i$; that is, for $i \in \{1, 2, \ldots, \ell\}$, $e_i = 1$ for $i \in \Gamma$ and $e_i = 0$ for $i \notin \Gamma$. 
\item \textbf{(``Push'')} Calvin constructs the set of $(u,r)$ that have encodings $\bx(u,r)=\Phi(u,r)$ that are \textit{close} to $\mbf{y}_1=y_1, \dots, y_\ell$.  Namely, Calvin constructs the set
	\begin{align}
	B_{\mbf{y}_1} = \{ (u,r) : d_H(\by_1, \bx_1(u,r))=\bar{p}n \},
	\end{align}
and selects an element $(\mess',\ind') \in B_{\mbf{y}_1}$ uniformly at random.
Calvin then considers the corresponding codeword $\bx'=\Phi(\mess',\ind')$.
Given the selected $\mbf{x}'$, for $i >
\ell$, if $x_i \ne x'_i$, Calvin sets $e_i$ equiprobably to $0$ or $1$ until
$\sum_{j = 1}^{i} e_j = p n$ or $i = n$. Note that, under our assumption
(w.l.o.g.) of uniform $\p_{\bU,\rR}$, the {\it a posteriori} distribution of
Alice's choice $(u,r)$ given $\by_1$ is also uniform in $B_{\by_1}$.
\end{itemize}





We start by proving the following technical lemma that we use in our proof.

\begin{lemma}
\label{lemma:sampling}
Let $V$ be a random variable on a discrete finite set $\mc{V}$ with entropy $H(V) \geq \lambda$, and let $V_1, V_2, \ldots, V_m$ be i.i.d. copies of $V$.  Then
	\begin{align}
	\Pr\left( \{V_i : i = 1, \ldots, m\} \ \textrm{are all distinct} \right)
	& \nonumber \\
	& \hspace{-0.7in}
	\ge
	\left( \frac{\lambda -  1 - \log m}{\log |\mc{V}|} \right)^{m-1}
	\label{eq:multselect}
	\end{align}
\end{lemma}

\begin{proof}
Fix $i \le m$ and a set $v_1, v_2, \ldots, v_i \in \mc{V}$.  Let $A_i = \{v_1, \ldots, v_i\}$ and let $W_i = \mbf{1}(V_{i+1} \in A_i)$, \sj{where $\mbf{1}(.)$ denotes the {\it indicator function}}.  We can write the distribution of $V$ as a mixture:
\sj{	\[
	\Pr\left[ V_{i+i} = v \right] 
		= \sum_{{j \in \{0,1\}}}\Pr[ W_i=j ] \cdot {\Pr[ V_{i+i}=v | W_i=j ]}
	\]
}
We can {\it bound from above} the entropy of $V$ as:
\sj{	\begin{align*}
	H(V_{i+1}) &\le H(V_{i+1} | W_i) + H(W_i) \\
	&= \sum_{{j \in \{0,1\}}} \Pr[ W_i=j ] H( V_{i+1} | W_i=j) + H(W_i)
	\end{align*}
}
Since conditioning reduces entropy and the support of \sj{$V_{i+1}$} conditioned on $W_i = 1$ is
at most $i$, we have
	\[
	\lambda \le 1 + \log i + \Pr[W_i=0] \log |\mc{V}|.
	\]
\sj{ Namely, 
$$\Pr[W_i=0] \ge \frac{\lambda - 1 - \log i}{\log|\mc{V}|} \ge  \frac{\lambda - 1 - \log m}{\log|\mc{V}|}.$$}  
\sj{But the event that each $V_i$ is distinct is equivalent to the event that for each $i \in \{2,\ldots,m\}$, $W_{i}$ is $0$. }
\end{proof}

To prove the upper bound, we now present a series of claims.  Let $\mbf{X}$ denote the random variable corresponding to Alice's input codeword and let $\mbf{Y}$ be the output of the channel.  Thus $\mbf{X}_1 \in \{0,1\}^{\ell}$ is Alice's input during the ``babble'' phase of length $\ell$ and $\mbf{X}_2$ is her input during the ``push'' phase; the randomness comes from the message $\bU$ and the stochastic encoding.  Similarly, $\mbf{Y}_1$ is the random variable corresponding to the $\ell$ bits received by Bob during the ``babble'' phase, and $\mbf{Y}_2$ the $\bl - \ell$ bits of the ``push'' phase.  Let $\mathsf{Adv}$ denote the ``babble-and-push'' adversarial strategy.

Let
	\begin{align*}
	A_0 = \{\mbf{y}_1: H(\bU |\mbf{Y}_1 = \mbf{y}_1) \ge \bl \e/4\},
	\end{align*}
where the entropy $H(\bU |\mbf{Y}_1 = \mbf{y}_1)$ is measured over the randomness of the encoder, \sj{the message, and any randomness in Calvin's action during the ``babble'' phase. Further, let the event $E_0$ be defined as}
	\begin{align}
	E_0 = \{ \mbf{Y}_1 \in A_0 \}.
	\label{eq:err0}
	\end{align}

\begin{claim}
\label{claim:entropy1}
\sj{For the ``babble-and-push'' attack $\mathsf{Adv}$,
	\begin{align}
	\Pr_{\mathsf{Adv}}( E_0 ) \ge \e/4.
	\end{align}}
\end{claim}

\begin{proof}
By the data processing inequality (\sj{$\bU \rightarrow \mbf{X}_1 \rightarrow \mbf{Y}_1$ form a Markov chain and hence $I( \bU ; \mbf{Y}_1 ) \leq I( \mbf{X}_1 ; \mbf{Y}_1)$}), and the choice of Calvin's strategy, we have
	\begin{align*}
	I( \bU ; \mbf{Y}_1 ) &\le I( \mbf{X}_1 ; \mbf{Y}_1) \\
	&\le \ell ( 1 - H( \bar{p} n/\ell ) ) \\
	& = (\alpha \bl + \e \bl/2) 
		\left ( 1 - H\left( \frac{\bar{p}}{ \alpha + \e/2 } \right) \right)
	\end{align*}
Therefore
	\begin{align*}
	H( \bU | \mbf{Y}_1 ) &\ge H(\bU) - n (\alpha + \e/2) 
		\left( 1 - H\left( \frac{\bar{p}}{ \alpha + \e/2 } \right) \right) \\
	&\ge n \left( \e 
			+ \alpha \left( 1 - H\left( \frac{\bar{p}}{\alpha} \right) \right) 
			\right) 
		\\
		&\hspace{0.5in}
		- n (\alpha + \e/2) 
		\left( 1 - H\left( \frac{\bar{p}}{ \alpha + \e/2 } \right) \right)\\
	&= n\e/2 + n\Bigg( (\alpha + \e/2)H\left( \frac{\bar{p}}{ \alpha + \e/2 } \right) 
		\\
		&\hspace{0.8in}
		- \alpha  H\left( \frac{\bar{p}}{\alpha}  \right)
			\Bigg) \\
	&\geq \bl \e/2.
	\end{align*}
\sj{Here the first inequality follows from the definition of conditional entropy, the second from the assumption underlying this proof by contradiction that $n R$ (and hence $H(\bU)$) violates \eqref{eq:assum_start}--\eqref{eq:assum_end}, and the third from the fact that the function $\alpha  H\left( {\bar{p}}/{\alpha} \right)$ is monotonically increasing in $\alpha$ since the function's derivative with respect to $\alpha$ equals $\log (\alpha/(\alpha-\bar{p}))$ which is always positive.}
Thus the expected value of $H(\bU | \mbf{Y}_1 = \mbf{y}_1)$ over $\mbf{y}_1$ is at least $n \e/2$, and the maximum value of $H(\bU | \mbf{Y}_1 = \mbf{y}_1)$ is $n R$.  \sj{Applying the Markov inequality to the random variable $nR - H(\bU | \mbf{Y}_1 = \mbf{y}_1)$, we see that
	\begin{align*}
	\Pr \left [ nR - H(\bU | \mbf{Y}_1 = \mbf{y}_1) > n R - n \e/4 \right ] &
		\\
		&\hspace{-1in}
		< \frac{nR - n \e/2}{ nR - n\e /4 } \\
		&\hspace{-1in}
		= 1 - \frac{\e/4}{ R - \e /4 },
	\end{align*}
and hence 
	\begin{align*}
	\Pr \left [   H(\bU | \mbf{Y}_1 = \mbf{y}_1) \ge   n \e/4 \right ] & \ge \frac{\e/4}{ R - \e /4 }.
	\end{align*}
Using the fact that $R \le 1$ yields the result.
}
\end{proof}

Now consider drawing $m$ pairs $(U_i,R_i)$ from $B_{\by_1}$ i.i.d.
$\sim \p_{\bU,\rR|\by_1}$ (which happens to be uniform).
Note that the marginal distribution of $U_i$ is also i.i.d. $\sim \p_{\bU|\by_1}$, which
is not necessarily uniform.  Let
	\begin{align}
	E_1 = \left\{ \{U_1, U_2, \ldots, U_m\} \ \textrm{are all distinct} \right\}.
	\end{align}	

\begin{claim}
\label{claim:uniqueness}
Let $\p_{\bU|_{\by_1}}$ be the conditional distribution of $\bU$ given
$\mbf{y}_1$ under $\mathsf{Adv}$.  Let $U_1, U_2, \ldots,
U_m$ be $m$ random variables drawn i.i.d. according to $\p_{\bU|_{\by_1}}$.
Then \sj{for large enough $n$},
	\sj{\begin{align}
	\Pr( E_1 ~|~ E_0 ) \ge (\e/5)^{m-1}.
	\end{align}}
\end{claim}

\begin{proof}
The proof follows from Claim~\ref{claim:entropy1} and Lemma~\ref{lemma:sampling} by using $\lambda = n \e/4$, \sj{ $V = \bU$ }, and the fact that there are at most $2^n$ messages, \sj{so $|\mc{V}| \le 2^n$}.  \sj{ The lower bound in \eqref{eq:multselect} then becomes $\left(\frac{ n \e/4 - 1 - \log m}{ n } \right)^{m-1}$.  For fixed $m$ there exists a sufficiently large $n$  such that  $\e/4 - (1 + \log m)/n > \e/5$.}
\end{proof}

The preceding two claims establish a lower bound on the probability that $\mbf{Y}_1$ takes a value such that the distribution of the message $\bU$ conditioned on $\mbf{Y}_1$ has sufficient entropy.
For such values $\by_1$ of $\mbf{Y}_1$, we now use the fact that Alice's pair $(\mess,\ind)$ is uniform in $B_{\by_1}$ to analyze the probability that Calvin's ``push'' attack succeeds.  Let $\bU'$ and $\mbf{X}'$ denote the random choice of Calvin's message and codeword in the ``push'' phase.  \sj{We show that the following two events}  occur with probability bounded away from zero:
\sj{	\begin{align}
	E_2 &= \{ \bU' \ne \bU \} \\
	E_3 &= \left\{ d_H\left( \mbf{X}_2, \mbf{X}'_2 \right) \le 2(p-\bp)\bl- \e \bl/8 \right\}
	\end{align}
}
The first event is that Calvin chooses a different message than Alice and the second is that he chooses a codeword that is close enough to Alice's. \sj{The occurrence of the first event ensures that the codeword Calvin chooses to try to confuse Bob into thinking might have been transmitted corresponds to a message $u'$ different than Alice's actual \bikd{message} $u$. The occurrence of the second event ensures that the two codewords chosen ($\bx_2$ chosen by Alice, and $\bx_2'$ by Calvin) are ``close enough'' for Calvin to be able to push Bob's received codeword halfway between $\bx_2$ and $\bx_2'$.}


\begin{claim}
For the ``babble-and-push'' attack $\mathsf{Adv}$,
	\begin{align}
	\Pr_{\mathsf{Adv}}\left( E_2 \ \textrm{and}\ E_3 ~|~ E_0 \right) \ge \e^{O(1/\e)}.
	\end{align}
\label{claim:hamm2}
\end{claim}

\begin{proof}
Conditioned on $E_0$, 
the realization $\by_1$ satisfies $H(\mbf{U}|\mbf{Y}_1=\by_1) \geq \e \bl/4$.  
We first use Claim \ref{claim:uniqueness} to lower bound the probability that $E_2$ holds.
First consider randomly sampling a set of \sj{mutually independent} pairs $S = \{(\mess_i,\ind_i) : i \in [m] \}$ uniformly from $B_{\by_1}$, and let $\mbf{X}^i$ be the codeword for $(\mess_i,\ind_i)$.  

Claim \ref{claim:uniqueness} shows that with probability at least \sj{$(\e/5)^{m-1}$}, all the messages in $S$ are distinct.  In particular, this shows that
	\sj{\[
	\Pr_{\mathsf{Adv}}\left( E_2 ~|~ E_0 \right) \ge (\e/5).
	\]}
Turning to $E_3$, applying Claim \ref{claim:uniqueness} for general $m$ shows that the probability that $m$ draws from the conditional distribution $\p_{\bU|_{\by_1}}$ yield unique messages is lower bounded by \sj{$(\e/5)^{m-1}$}.  Plotkin's bound~\cite{Plo:60} \sj{(reprised in Theorem~\ref{thm:plo})} shows that there do not exist binary error-correcting codes of block-length $n - \ell$ and minimum distance $d$ with more than $\frac{2d}{2d - (\bl - \ell)}$ codewords.  Setting $m = 17/\e$, this bound implies that with probability at least \sj{ $(\e/5)^{m-1}$} there must exist codewords $\bx, \bx'$ corresponding to $(\mess,\ind)$ and \sj{$(\mess',\ind')$ respectively (with $\mess \neq \mess'$)} within a distance $d$ that satisfies
	\[
	\frac{17}{\e} \leq \frac{2d}{2d - (\bl - \ell)} 
	\]
Solving for $d$ and using $\ell = (1 - 4(p - \bp) + \e/2) n$ shows that $d$ satisfies
\sj{	\[
	d \leq 2(p-\bp) \bl\frac{17}{17+\e} - \frac{\e\bl}{4}\frac{17}{17+\e} < 2(p-\bp) \bl - \e \bl/8.
	\]}
Let $\Delta = 2(p-\bp) \bl - \e \bl/8$\sj{.}

Let $\gamma$ be the fraction of pairs $(\mess,\ind)$ and $(\mess',\ind')$ in $B_{\by_1}$
that satisfy $E_2$ and $E_3$.   We would like to lower bound $\gamma$. 
A union bound shows that the probability over the selection
of $S$ gives the upper bound
\ml{
	\begin{align}
	\label{eq:new1}
	\Pr\left( \bigcup_{S}
		\{ d_H(\mbf{X}^i,\mbf{X}^j) < \Delta \} \ \textbf{and} \  \{\bU^
		i \ne \bU^j\} \  \right) \le m^2 \gamma.
	\end{align}
}

%

However, the earlier argument shows that by selecting $m = 17/\e$ pairs in $S$,
we get a lower bound of $(\e/5)^{m-1}$ on 
the probability that (a) all $\{\bU^i\}$ are distinct, and (b) at least one pair $\mbf{X}^i$, $\mbf{X}^j$ has distance less that $\Delta$:
	\begin{align}
	\Pr\left( \{ \textrm{all $\bU^i \in S$ are distinct} \} \ 
	\textrm{and}\   
	\bigcup_{S} \{d_H(\mbf{X}^i,\mbf{X}^j) < \Delta \} \  \right) 
	& \nonumber \\
	&\hspace{-1.5in}
	\ge (\e/5)^{m-1}.
	\label{eq:new2}
	\end{align}


\ml{
As the event analyzed in Equation~(\ref{eq:new1}) includes that analyzed in Equation~(\ref{eq:new2}),    
we have that 
}
	\sj{\[
	\gamma \ge \frac{1}{m^2} \left(\frac{\e}{5}\right)^{m-1} 
	= \frac{17^2}{\e^2} \left(\frac{\e}{5}\right)^{17/\e - 1}.
	\]}
%
\ml{Therefore, by the definition of $\gamma$, we conclude our assertion.}
\end{proof}


The next step is to show that Calvin does not ``run out'' of bit flips during the second ``push'' phase of his attack.   This follows directly from Chernoff's bound~\cite{AloS:11}.   

\sj{We now analyze Calvin's action during the ``push'' phase. This action can be viewed as being equivalent to the following two stages. In the first stage, $d_H(\mbf{X}_2, \mbf{X}_2')$ bits are drawn i.i.d. Bernoulli-$(1/2)$ -- these bits comprise the {\it intended error vector} $\mbf{\hat{e}}$. However, Calvin may not have the power to impose this intended vector in the push phase if the weight of $\mbf{\hat{e}}$ is too large. In general, the bit-flips in Calvin's actual error vector $\mbf{{e}_2}$ correspond to the components of $\mbf{\hat{e}}$ up to the point that he runs out of his bit-budget. }

\ml{Let $d$ be the distance between the $\mbf{X}_2$ chosen by Alice and $\mbf{X}'_2$ chosen by Calvin and let}
\sj{the event $E_4$ be defined as}
\ml{
	\begin{align}
	E_4 = \left \{  wt_H(\mbf{\hat{e}}) \in 
		\left( \frac{d}{2}-\frac{\e \bl}{16},
			\frac{d}{2}+\frac{\e \bl}{16}\right) 
		\right \}.
	\label{eq:no_bit_flips2}
	\end{align}
}

\begin{claim}
\label{claim:chernoff}
For the ``babble-and-push'' attack $\mathsf{Adv}$,
	\begin{align}
	\Pr_{\mathsf{Adv}}\left( E_4 ~|~ E_2, E_3 \right) \ge 1 - 2^{-\Omega(\e^2 \bl)}.
	\end{align}
\label{claim:budget2}
\end{claim}


\begin{proof}
\ml{As $d$ is} the distance between the $\mbf{X}_2$ chosen by Alice and $\mbf{X}'_2$ chosen by Calvin, without any constraint, Calvin would flip $d/2$ locations in expectation.   Conditioned on $E_2$ and $E_3$, \sj{we have the following upper bound:}
	\[
	\frac{d}{2} \leq (p-\bp) \bl - \e \bl/16.
	\]
Assume that $d/2 = (p-\bp)\bl-\e \bl/16$ (for smaller values of $d$ the bound is only tighter).
By Chernoff's bound~\cite{AloS:11}, the probability that the number of \bikd{bit flips in $\hat{e}$ (i.e., the Hamming weight of $\hat{e}$)} deviates from the expectation by more than $\e \bl/16$ is at most $2^{-\Omega(\e^2 n)}$.
\end{proof}

Note that the number of bit flips in the first phase of the algorithm is exactly $\bp \bl$, 
\ml{and thus Claim~\ref{claim:chernoff} implies that with high probability the total number of bit flips in $\mbf{\hat{e}}$ in the second phase will not exceed $\frac{d}{2} + \frac{\e \bl}{16} \leq (p -\bp) \bl$ 
and will not be significantly less than that expected ({\it i.e.}, less than $\frac{d}{2} - \frac{\e \bl}{16}$ -- in this case Bob might be able to conclude that $\mbf{X}'_2$ was not transmitted). 
If this is not the case, our analysis assumes Calvin (in the worst case for him) fails to jam Alice's transmission to Bob.}


\begin{theorem}
\label{claim:sym}
For any code with stochastic encoding of rate $R = C + \e$, under Calvin's ``babble-and-push'' strategy the average error probability $\bar{\varepsilon}$ is lower bounded by $\e^{O(1/\e)}$.
\end{theorem}

\begin{proof}
The main idea \sj{behind the proof of our outer bound} is that conditioned on events $E_0$, $E_2$, $E_3$, and $E_4$, (whose probabilities of occurrence are analyzed in Claims~\ref{claim:entropy1}, \ref{claim:hamm2}, and~\ref{claim:budget2}), Calvin can ``symmetrize'' the channel \cite{CsiszarN:88positivity}. \sj{That is, Calvin can choose to inject bit-flips in a manner so that Bob is unable to distinguish between two possible codewords $\bx$ and $\bx'$ (corresponding to different messages $u$ and $u'$) transmitted by Alice. Calvin does this by ensuring (with probability bounded away from zero) that the codeword received by Bob, $\by$, {\ml is likely to equal either} $\bx+\be$ or $\bx'+\be'$ for two valid pairs $(\bx,\be)$ and $(\bx',\be')$ of transmitted codewords and bit-flip vectors.}  

Let $(\mess,\ind)$ denote the message and randomness of Alice, $\mbf{y}_1$ be the received codeword in the ``babble'' phase, and $(\mess',\ind')$ be the message and randomness chosen by Calvin for the ``push'' phase.  Let $\p( \by_1,\mess,\ind,\mess',\ind' )$ be the joint distribution of these variables under Alice's uniform choice of $(\mess,\ind)$ and Calvin's attack.  For each $\by$, let $\p(\by | \by_1,\mess,\ind,\mess',\ind')$ be the conditional distribution of $\by$ under Calvin's attack.

The error probability can be written as
	\begin{align*}
	\bar{\varepsilon} &= \sum_{\by_1,\mess,\ind,\mess',\ind'} \p( \by_1,\mess,\ind,\mess',\ind' ) 
	\\
	&\hspace{0.5in}
	\sum_{\by_2} \p(\by | \by_1,\mess,\ind,\mess',\ind') \Pr(\Psi(\by) \ne \mess).
	\end{align*}

Let $\mc{F}$ be the set of tuples $(\by_1,\mess,\ind,\mess',\ind')$ satisfying events $E_0$, $E_2$, and $E_3$.  Claims \ref{claim:entropy1} and \ref{claim:hamm2} show that \ml{$\p( \mc{F} ) \ge (\e/4) \cdot \e^{O(1/\e)}$}.  For $(\by_1,\mess,\ind,\mess',\ind') \in \mc{F}$, we have that $u \ne u'$, and that $\bx_2(\mess,\ind)$ and $\bx_2(\mess',\ind')$ are sufficiently close.

\sj{
Assuming $E_4$ holds, if $\by_2$ results from $\mbf{x}_2$ via $\mbf{\hat{e}}$, then $\by_2$ may also have resulted from $\mbf{x}_2'$ via $\mbf{\hat{e}}^C$ (the binary complement of $\mbf{\hat{e}}$). Since $\mbf{\hat{e}}$ is generated via i.i.d. Bernoulli-$(1/2)$ components, $\mbf{\hat{e}}$ and $\mbf{\hat{e}}^C$ have the same probability.}

Thus the conditional distribution is symmetric:
	\begin{align}
	\p(\by | \by_1,\mess,\ind,\mess',\ind') = \p(\by | \by_1,\mess',\ind',\mess,\ind).
	\end{align}
Then for $(\mbf{y}_1,\mess,\ind,\mess',\ind') \in \mc{F}$, by Claim \ref{claim:budget2},
	\begin{align*}
	\sum_{ \by_2 \in \mc{G} } \p(\by_2 | \by_1,\mess,\ind,\mess',\ind') 
	&
	\ge 1 - 2^{-\Omega(\e^2 n)}.
	\end{align*}
	
Now, returning to the overall error probability, let $\p(\by_1)$ be the unconditional probability of Bob receiving $\by_1$ in the ``babble'' phase, where the probability is taken over Alice's uniform choice of $(u,r)$ and Calvin's random babble $\mbf{e}_1$.  
\bikd{Since the a-posteriori distribution of $(u,r)$ and $(u',r')$ given
$\by_1$ are independent and both uniform in $B_{\by_1}$, the joint
distribution can be written as
	\begin{align*}
	\p( \by_1,\mess,\ind,\mess',\ind' )
		&= \p(\by_1) \cdot \frac{1}{|B_{\by_1}|^2} \\
& = 
\p( \by_1,\mess',\ind',\mess,\ind).
	\end{align*}
}
Recall that for any $\by_2\in \mc{G}$,
\begin{align}
\p (\by_2 | \by_1,\mess,\ind,\mess',\ind') = \p (\by_2 | \by_1,\mess',\ind',\mess,\ind). \nonumber
\end{align}
Thus,
	\begin{align}
	2 \bar{\varepsilon} 
        &\ge \sum_{\mc{F}} \p( \by_1,\mess,\ind,\mess',\ind' )  
        		\nonumber \\
		&\hspace{0.4in}
		\Bigg(\sum_{\by_2\in \mc{G}} \p(\by_2 | \by_1,\mess,\ind,\mess',\ind')
			\Pr(\Psi(\by_1,\by_2) \ne \mess )  
		\nonumber \\
		&\hspace{0.5in} 
		+ \sum_{\by_2 \in \mc{G}} \p(\by_2 | \by_1,\mess',\ind',\mess,\ind) \Pr(\Psi(\by_1,\by_2) \ne \mess' )\Bigg)
       \nonumber \\
	&\geq \sum_{\mc{F}} \p( \by_1,\mess,\ind,\mess',\ind' ) 
		\sum_{\by_2\in \mc{G}} \p(\by_2 | \by_1,\mess,\ind,\mess',\ind')
		\nonumber \\
		&\hspace{0,6in}
		\left( \Pr(\Psi(\by_1,\by_2) \ne \mess ) + \Pr(\Psi(\by_1,\by_2) \ne \mess' )\right)
	\nonumber \\
        &\geq \sum_{\mc{F}} \p( \by_1,\mess,\ind,\mess',\ind' ) \sum_{\by_2\in \mc{G}} \p(\by_2 | \by_1,\mess,\ind,\mess',\ind')
        \\
	&\ge \e/4 \cdot \e^{O(1/\e)} \cdot \left( 1 - 2^{-\Omega(\e^2 n)} \right). \mbox{\hspace{1.34in}}	
		\label{eq:endbound} \nonumber
	\end{align}

\end{proof}

Our analysis implies a refined statement of Theorem~\ref{the:det}. Namely, let $c$ be a sufficiently large constant.
For any block-length $\bl$, any $\e>0$ and any encoding/decoding scheme of Alice and Bob of rate $\left(C(p)+ \sqrt{\frac{c}{n}}\right) +\e$, Calvin can cause a decoding error \sj{probability}
of at least $\e^{O(1/\e)}$.

\section{Improved bounds for deterministic codes \label{sec:det}}

We now present an alternative analysis for the case of deterministic encoding.  Without loss of generality, we assume each codeword corresponds to a unique message in $\cU$ so there are $2^{\rate \bl}$ distinct \sj{equiprobable} codewords in $\{ \bx(\mess) : \mess \in \cU \}$, \sj{with a unique codeword for each message}.  
The attack is the same as in Section \ref{sec:main}. Apart from the simpler proof, the analysis below gives a decoding error $\bar{\e}$ proportional to $\e$, which improves over the decoding error presented for stochastic \sj{encoding} appearing in the body of this work.

\ml{
Using the notation of Section \ref{sec:main}, for any vector $\by_1$ consider the set 
\begin{align}
B_{\by_1} = \{ (\bx,\be_1) : \bx_1+\be_1=\by_1,\  \be_1=\bar{p}n\}.
\end{align}
Here, $\be_1$ represents the potential error vector that Calvin imposes in the first stage of its attack on the transmitted codeword $\bx$.
Notice that the set $B_{\by_1}$ defined above is analogous to the set $B_{\by_1}$ defined in Section~\ref{sec:main}. Namely, for any message $u$ a pair $(u,r) \in B_{\by_1}$ in Section~\ref{sec:main} corresponds to a pair $(\bx(u),\be_1) \in B_{\by_1}$ defined above. We note that in the definition above $B_{\by_1} \cap B_{\by'_1}=\phi$ for $\by_1 \ne \by'_1$ as we assume all codewords to be distinct.
}

\begin{claim}
With probability at least $1/2$ over the codeword $\bx$ sent by Alice and the actions of Calvin in the first stage of his attack, the set $B_{\by_1}$ is of size at least $2^{\e \bl/4}/2$.
\label{claim:suff_set}
\end{claim}

\begin{proof}
The proof is obtained by the following counting argument.
The number of possible sets $B_{\by_1}$ is exactly $2^\ell=2^{\alpha \bl + \e \bl/2}$.
The number of pairs \ml{$(\bx,\be_1)$} for a codeword $\bx$ and an error vector \ml{$\be_1$} (to be applied in the first stage by Calvin) is
	\begin{align*}
2^{R \bl}{{\ell}\choose{\bp \bl}} 
	&\geq 2^{R \bl}{{\alpha \bl }\choose{\bp \bl}} \\
	&\geq 2^{\rate \bl} \cdot 2^{\alpha \bl H(\bp/\alpha) - \e \bl/4} \\
	&\geq 2^{\alpha \bl+ 3\e \bl/4}.
	\end{align*}
Here \sj{the first inequality follows from the fact that $\ell \leq \alpha \bl$, the second inequality from the standard bound $2^{nH(k/n)}/(n+1) \leq {{n}\choose{k}}$ (for instance~\cite[Theorem 11.1.3]{CT06}) and the fact that $\bl$ is sufficiently large with respect to $1/\e$ and hence $2^{-\e \bl/4}$ is smaller than any polynomial in $1/(\bl+1)$, and the third inequality from the starting assumption that $R$ is at least $\e + \alpha(1-H(\bp/\alpha))$.}
Thus, the {\em average} size of a set $B_{\by_1}$ is at least $2^{\e \bl/4}$.
Consider all the sets $B_{\by_1}$ of size less than half the average $2^{\e \bl/4}/2$.
The total number of codewords in the union of these sets is at most
$$
2^\ell \cdot 2^{\e \bl/4}/2 \leq 2^{R \bl}{{\ell}\choose{\bp \bl}} \cdot \frac{1}{2},
$$
which is half the number of \ml{$(\bx,\be_1)$} pairs.
As each pair is chosen with the same probability, we conclude that with probability at least $1/2$ the pair \ml{$(\bx,\be_1)$} appears in a set $B_{\by_1}$ which is of size at least  $2^{\e \bl/4}/2$.
This completes the proof of our assertion.
\end{proof}

We now show that Claim~\ref{claim:suff_set} above implies that the transmitted codeword $\bx$ and the codeword $\bx'$ chosen by Calvin are distinct and of {\em small} Hamming distance apart
with a positive probability (independent of $\bl$).

\begin{claim}
Conditioned on Claim~\ref{claim:suff_set}, with probability at least $\frac{\e}{64}$, $\bx \neq \bx'$ and $d_H(\bx_2,\bx_2') < 2(p-\bp)n- \e n/8$.
\label{claim:hamm2new}
\end{claim}

\begin{proof}
Consider the undirected graph $\Graph=(\cV,\cE)$ in which the vertex set $\cV$ consists of the set $B_{\by_1}$ and  two nodes $\bx$ and $\bx'$ are connected by an edge if 
$d_H(\bx_2,\bx'_2) \leq d=2(p-\bp)n- \e n/8$.  The set of codewords defined by the suffixes of an independent set $\cI$ in $\Graph$ correspond\sj{s} to a binary error-correcting code with block-length \ml{$\bl-\ell = 4(p-\bp) \bl -\e \bl/2$} of size $|\cI|$ and minimum distance $d$.	

By Plotkin's bound~\cite{Plo:60} \sj{(reprised in Theorem~\ref{thm:plo})} there do not exist binary error correcting codes with more than $\frac{2d}{2d-(4(p-\bp)\bl - \e \bl/2)}+1$ codewords.
Thus $\cI$, any maximal independent set in $\Graph$, must satisfy
\begin{eqnarray}
|\cI| 
& \leq & \frac{2(2(p-\bp)n-\e n/8)}{2(2(p-\bp)n -\e n/8)-4(p-\bp)n + \e n/2}+1 \nonumber \\
& = & \frac{16(p-\bp)}{\e} \leq \frac{16}{\e}.
\label{eq:ind2}
\end{eqnarray}

By Tur\'{a}n's theorem \cite{turan}, any undirected graph $\Graph$ \sj{on $|\cV|$ vertices} and average degree $\Delta$ has an independent set of size at least $|\cV|/(\Delta+1)$. This, along with (\ref{eq:ind2}) implies that the average degree of our graph $\Graph$ satisfies
$$
\frac{|\cV|}{\Delta+1}\leq |\cI| \leq \frac{16}{\e}.
$$
This in turn implies that
$$
\Delta \geq \frac{\e |\cV|}{16 } -1 \geq \frac{\e |\cV|}{32}.
$$
The second inequality holds for our setting of $\bl$,
since $|\cV|$ is of size at least $2^{\e \bl/4}$.
To summarize the above discussion, we have shown that our graph $G$ has {\em large} average degree of size $\Delta \geq \frac{\e |\cV|}{32p}$. We now use this fact to analyze Calvin's attack.

By the definition of deterministic codes, any \sj{valid codeword in $\cX^n$} is transmitted with equal probability.
Also, by definition both $\bx$ (the transmitted codeword) and $\bx'$ (the codeword chosen by Calvin) are in $\cV=B_{\by_1}$. 
Hence both $\bx$ and $\bx'$ are uniform in $B_{\by_1}$. This implies that with probability $|\cE|/|\cV|^2$ the nodes corresponding to codewords $\bx$ and $\bx'$ are distinct and connected by an edge in $\Graph$.
This in turn implies that with probability $|\cE|/|\cV|^2$, $\bx \neq \bx'$ and $d_H(\bx,\bx') < 2(p-\bp)\bl- \e \bl/8$, as required.
Now
$$
\frac{|\cE|}{|\cV|^2} = \frac{\Delta|\cV|}{2|\cV|^2} \geq  \frac{\e}{64}.
$$
\end{proof}

The preceding claims provide the same guarantees as Claim \ref{claim:hamm2} appearing in the body of the paper, and so Claim \ref{claim:budget2} follows. Namely, w.h.p., Calvin does not ``run out'' of his budget of $p\bl$ bit flips.
We conclude by proving that given the analysis above Bob cannot
distinguish between the case in which $\bx$ or $\bx'$ were transmitted, using a
similar symmetrization argument.

\begin{theorem} \label{thm:deterministic}
For any code with deterministic encoding and decoding of \sj{rate} $R = C + \e$, under Calvin's ``babble-and-push'' strategy the average error probability $\bar{\varepsilon}$ is lower bounded by $\frac{\e}{256} \left( 1 - 2^{-\Omega(\e^2 n)} \right)$
\end{theorem}

\begin{proof}
Let $u$ be the message chosen by Alice, $u'$ be the message chosen by Calvin, $\p(\by_1,u,u')$ be the joint distribution of the output during the ``babble'' phase and these two messages, and $\p(\by|\by_1,u,u')$ be the conditional distribution of the output on the result of the ``babble'' phase.

Let $\mc{G}'$ be the set of $\by_2$ such that Claim \ref{claim:budget2} is satisfied for $\bx_2(\mess)$.  As in the arguments of Theorem \ref{claim:sym}, Calvin's attack is symmetric, so that
	\[
	\p(\by_1, u, u') = \p(\by_1,u',u),
	\]
and therefore we have
	\begin{align*}
	\sum_{ \by_2 \in \mc{G}'} \p(\by | \by_1,\mess,\mess')
	\ge 1 - 2^{-\Omega(\e^2 n)}.
	\end{align*}

Let $\mc{F}$ be the set of tuples $(\by_1,\mess,\mess')$ satisfying Claim \ref{claim:suff_set} and Claim \ref{claim:hamm2new}. Following the analysis in Theorem \ref{claim:sym}, from \eqref{eq:endbound} and applying Claims \ref{claim:suff_set} and Claim \ref{claim:hamm2new}, we have
	\begin{align*}
	2 \bar{\varepsilon} &\ge \sum_{\mc{F}} \p(\by_1, u, u') \left( 1 - 2^{-\Omega(\e^2 n)} \right) \\
	&\ge \frac{\e}{128} \left( 1 - 2^{-\Omega(\e^2 n)} \right).
	\end{align*}
Dividing both sides by 2 yields the result.
\end{proof}

\section{Concluding remarks}\label{sec:conc}

In this paper we presented a novel upper bound on the rates achievable on binary additive channels with a causal adversary.  This model is weaker than the traditional worst-case error model studied in coding theory, but is stronger than an i.i.d. model for the noise.  Indeed, our results show the binary symmetric channel capacity $1 - H(p)$ is not achievable against causal adversaries.  By contrast, previous work shows that a delay of $D n$ (with $D$ a positive constant in $(0,1]$) for the adversary allows Alice and Bob to communicate at rate $1 - H(p)$.  Thus the causal adversary is strictly more powerful than the delayed adversary (which in turn is no stronger than i.i.d. noise).

To show our bound we demonstrated a new ``babble-and-push'' attack.  The adversary increases the uncertainty at the decoder during the ``babble'' phase, enabling it to choose an alternative codeword during the ``push'' phase.  The ``push'' phase succeeds because the adversary can effectively symmetrize the channel.  We demonstrate that the upper bound presented herein holds against arbitrary codes, rather than simply against deterministic codes, as is common in the coding theory literature.  Since our analysis pertains to adversarial jamming rather than random noise, the proof techniques presented may be of independent interest in the more general setting of AVCs.

\appendices

\section{The minimization in Theorem~\ref{the:det}}\label{app:min}
Let us denote $\bar{p}$ by $x$, and write the bound as a function of $x$
as
$$
f(x) = (1-4p+4x)\left(1-H\left(\frac{x}{1-4p+4x}\right)\right).
$$
So
\begin{align*}
f'(x)
	& = 4 - \frac{d}{dx}\Bigg\{ (1-4p+4x) 
		\\ 
		&\hspace{0.9in}
		\Bigg[ - \frac{x}{1-4p+4x} \log\frac{x}{1-4p+4x}  
		\\ 
		&\hspace{1.1in}
			- \frac{1-4p+3x}{1-4p+4x}\log\frac{1-4p+3x}{1-4p+4x}
			\Bigg]
		\Bigg\}  \\
& = 4 - \frac{d}{dx} \Bigg\{ - x\log\frac{x}{1-4p+4x} 
		\\ 
		&\hspace{0.9in}
		- (1-4p+3x)\log\frac{1-4p+3x}{1-4p+4x}\Bigg\}  \\
& = 4 - \Bigg\{ - \log\frac{x}{1-4p+4x} - \frac{(1-4p+4x) - 4x}{1-4p+4x}  
		\\ 
		&\hspace{0.9in}
		-3 \log\frac{1-4p+3x}{1-4p+4x} 
		\\
		&\hspace{0.9in}
		- \frac{3(1-4p+4x) - 4(1-4p+3x)}{1-4p+4x}
		\Bigg\}  \\
& = 4 + \left\{ \log\frac{x}{1-4p+4x} + 3\log\frac{1-4p+3x}{1-4p+4x}\right\}  \\
& = 4 + \log\frac{x(1-4p+3x)^3}{(1-4p+4x)^4}.
\end{align*}
First, we check for any roots of $f'(x)$ in $0\leq x\leq p$.
\begin{align*}
f'(x) = 0 & \\
&\hspace{-0.3in}  \Leftrightarrow \frac{x(1-4p+3x)^3}{(1-4p+4x)^4} = \frac{1}{16}  \\
&\hspace{-0.3in}  \Leftrightarrow ((1-4p+3x) +x)^4 = 16 x(1-4p+3x)^3  \\
&\hspace{-0.3in}  \Leftrightarrow (1-4p+3x)^4 - 12 x(1-4p+3x)^3  
	\\
	&\hspace{-0.1in} 
	+ 6x^2(1-4p+3x)^2 + 4x^3(1-4p+3x) + x^4 = 0.
\end{align*}
We now substitute, for brevity, $a = (1-4p+3x)/x$.
\begin{align}
& a^4 -12a^3 + 6a^2 + 4a + 1 = 0 \nonumber \\
& \Leftrightarrow (a-1)(a^3-11a^2-5a-1)= 0 \nonumber 
\end{align}
We now consider two cases.
If $p=0.25$, we have that $f(x)=0$ for $x=0$.
Thus setting $x=0$ will yield the minimum value for $p=1/4$.

For $p < 0.25$, we study the minimum value given $x>0$.
When $x>0$ and $p <0.25$ it holds that $1-4p >0$ and $a>3$.
Thus, for $f'(x)$ to be zero we require that 
$a^3-11a^2-5a-1= 0$,
which can be found via the general formula for cubic equations (for instance~\cite[Chap. $6$]{Dun:90}) to have one real solution and two complex conjugate solutions.
The only real solution is 
	\begin{align*}
	a_0 &= \frac{1}{3}\left (11+\sqrt[3]{1592+24\sqrt{33}} + \sqrt[3]{1592-24\sqrt{33}} \right )  \\
	&\simeq 11.4445,
	\end{align*}
giving
$x = (1-4p)/(a_0-3) \simeq (1-4p)/8.4445$. However, this value is greater than $p$ for $p < 1/(a_0+1) \simeq 1/12.4445$.
For $p \in [1/(a_0+1), 0.25]$, this solution is in the range $[0,p]$. 

Now we will see that $f'(x)$ is negative for $0< x< (1-4p)/(a_0-3)$.

For $p< 0.25$, and $0< x< (1-4p)/(a_0-3)$, we have
\begin{align*}
a = \frac{1-4p+3x}{x} >  a_0 \simeq 11.4445,
\end{align*}
so
	\begin{align*}
	a^3 - 11a^2 - 5a - 1  &\simeq (a-a_0)(a^2+0.4445a+0.087)>0,
	\end{align*}
but since $(a-1)(a^3-11a^2-5a-1) > 0$ we have $f'(x) < 0$. 
By the continuity of the objective function, $f(0) = \lim_{x\rightarrow 0} f(x)$. 
So, $f(x)$ is decreasing in $0\leq x< (1-4p)/{a_0-3}$, and
thus the optimum $\bar{p}$ is given
by
	\begin{align*}
	\bar{p} &= \min\left\{p, \frac{1-4p}{a_0-3}\right\} \\
	&= \min\left\{p, \frac{3(1-4p)}{\left (2+\sqrt[3]{1592+24\sqrt{33}} + \sqrt[3]{1592-24\sqrt{33}} \right )}\right\} \\
	&\simeq \min\left\{p, \frac{1-4p}{8.4445}\right\}.
	\end{align*}


\begin{IEEEbiographynophoto}{Bikash Kumar Dey} (S'00-M'04)
received his B.E. degree in Electronics and Telecommunication Engineering from Bengal Engineering College, Howrah, India, in 1996. He received his M.E. degree in Signal Processing and Ph.D. in Electrical Communication Engineering from the Indian Institute of Science in 1999 and 2003
respectively.

From August 1996 to June 1997, he worked at Wipro Infotech Global R\&D. In February 2003, he joined Hellosoft India Pvt. Ltd. as a Technical Member. In June 2003, he joined the International Institute of Information Technology, Hyderabad, India, as Assistant Professor. In May 2005, he joined the Department of Electrical Engineering of Indian Institute of Technology Bombay where he works as Associate Professor. His research interests include information theory, coding theory, and wireless communication.

He was awarded the Prof. I.S.N. Murthy Medal from IISc as the best M.E. student in the Department of Electrical Communication Engineering and Electrical Engineering for 1998-1999 and Alumni Medal for the best Ph.D. thesis in the Division of Electrical Sciences for 2003-2004.
\end{IEEEbiographynophoto}

\begin{IEEEbiographynophoto}{Sidharth Jaggi} (MÕ00) received the B.Tech. degree in electrical engineering from the Indian Institute of Technology, Bombay, in 2000, and the M.S. and
Ph.D. degrees in 2001 and 2006, respectively, from the California Institute of Technology, Pasadena.

He is an Assistant Professor in Information Engineering at the Chinese University of Hong Kong.  His primary research interests are in network coding, coding theory and information theory.
\end{IEEEbiographynophoto}

\begin{IEEEbiographynophoto}{Michael Langberg} (M'07) received his B.Sc. in mathematics and computer science from Tel-Aviv University in 1996, and his M.Sc. and Ph.D. in computer science from the Weizmann Institute of Science in 1998 and 2003 respectively. 

Between 2003 and 2006, he was a postdoctoral scholar in the Computer Science and Electrical Engineering departments at the California Institute of Technology.   He is an Associate Professor in the Mathematics and Computer Science department at the Open University of Israel.  His research interests include information theory, combinatorics, and algorithm design.
\end{IEEEbiographynophoto}

\begin{IEEEbiographynophoto}{Anand D.~Sarwate} (S'99--M'09) received B.S. degrees in electrical engineering and computer science and mathematics from the Massachusetts Institute of Technology (MIT), Cambridge, in 2002 and the M.S. and Ph.D. degrees in electrical engineering in 2005 and 2008, respectively, from the University of California, Berkeley.  

From 2008 to 2011 he was a postdoctoral researcher at the Information Theory and Applications Center at the University of California, San Diego.  He is a Research Assistant Professor at the Toyota Technological Institute at Chicago. His research interests include information theory, distributed signal processing, machine learning, and privacy.

Dr. Sarwate received the Laya and Jerome B. Wiesner Student Art Award from MIT, and the Samuel Silver Memorial Scholarship Award and Demetri Angelakos Memorial Achievement Award from the EECS Department at University of California at Berkeley. He was awarded an NDSEG Fellowship from 2002 to 2005. He is a member of Phi Beta Kappa and Eta Kappa Nu.
\end{IEEEbiographynophoto}


\begin{thebibliography}{10}
\providecommand{\url}[1]{#1}
\csname url@samestyle\endcsname
\providecommand{\newblock}{\relax}
\providecommand{\bibinfo}[2]{#2}
\providecommand{\BIBentrySTDinterwordspacing}{\spaceskip=0pt\relax}
\providecommand{\BIBentryALTinterwordstretchfactor}{4}
\providecommand{\BIBentryALTinterwordspacing}{\spaceskip=\fontdimen2\font plus
\BIBentryALTinterwordstretchfactor\fontdimen3\font minus
  \fontdimen4\font\relax}
\providecommand{\BIBforeignlanguage}[2]{{%
\expandafter\ifx\csname l@#1\endcsname\relax
\typeout{** WARNING: IEEEtran.bst: No hyphenation pattern has been}%
\typeout{** loaded for the language `#1'. Using the pattern for}%
\typeout{** the default language instead.}%
\else
\language=\csname l@#1\endcsname
\fi
#2}}
\providecommand{\BIBdecl}{\relax}
\BIBdecl

\bibitem{isit2012}
B.~K. Dey, S.~Jaggi, M.~Langberg, and A.~D.Sarwate, ``Improved upper bounds on
  the capacity of binary channels with causal adversaries,'' in
  \emph{\sj{Proceedings of the International Symposium on Information Theory
  (ISIT)}}, 2012, pp. \sj{681--685}.

\bibitem{McEliece77}
R.~J. McEliece, E.~R. Rodemich, H.~Rumsey, Jr., and L.~R. Welch, ``New upper
  bounds on the rate of a code via the {D}elsarte-{M}ac{W}illiams
  inequalities,'' \emph{IEEE Trans. Information Theory}, vol. IT-23, no.~2, pp.
  157--166, 1977.

\bibitem{Gil52}
E.~N. Gilbert, ``A comparison of signalling alphabets,'' \emph{Bell Systems
  Technical Journal}, vol.~31, pp. 504--522, 1952.

\bibitem{Var57}
R.~R. Varshamov, ``Estimate of the number of signals in error correcting
  codes,'' \emph{Dokl. Acad. Nauk}, vol. 117, pp. 739--741, 1957.

\bibitem{TsaVN:07}
\sj{M. A. Tsfasman and S. G. Vl{\u{a}}dut and D. Nogin}, \emph{\sj{Algebraic
  geometric codes: basic notions}}, ser. \sj{Mathematical Surveys and
  Monographs}.\hskip 1em plus 0.5em minus 0.4em\relax \sj{American Mathematical
  Society}, \sj{2007}, vol. \sj{139}.

\bibitem{Singleton:64}
R.~Singleton, ``Maximum distance q-nary codes,'' \emph{IEEE Transactions on
  Information Theory}, vol.~10, no.~2, pp. 116--118, 1964.

\bibitem{ReedSolomon:60}
I.~S. Reed and G.~Solomon, ``Polynomial codes over certain finite fields,''
  \emph{Journal of the Society for Industrial and Applied Mathematics (SIAM)},
  vol.~8, no.~2, pp. 300--304, 1960.

\bibitem{Shannon48}
C.~E. Shannon, ``A mathematical theory of communication,'' \emph{The Bell
  System Technical Journal}, vol.~27, pp. 379--423,623--656, July, October
  1948.

\bibitem{HL11}
\sj{I. Haviv} and M.~Langberg, ``Beating the {Gilbert-Varshamov} bound for
  online channels,'' in \emph{Proceedings of the 2011 International Symposium
  on Information Theory}, St. Petersburg, Russia, 2011, pp. 1297--1301.

\bibitem{BlackwellBT:60random}
D.~Blackwell, L.~Breiman, and A.~Thomasian, ``The capacities of certain channel
  classes under random coding,'' \emph{Annals of Mathematical Statistics},
  vol.~31, no.~3, pp. 558--567, 1960.

\bibitem{AhlswedeW70:binary}
R.~Ahlswede and J.~Wolfowitz, ``The capacity of a channel with arbitrarily
  varying channel probability functions and binary output alphabet,''
  \emph{{Zeitschrift f\"{u}r Wahrscheinlichkeit und verwandte Gebiete}},
  vol.~15, no.~3, pp. 186--194, 1970.

\bibitem{Ahlswede:73fback}
R.~Ahlswede, ``Channels with arbitrarily varying channel probability functions
  in the presence of noiseless feedback,'' \emph{{Zeitschrift f\"{u}r
  Wahrscheinlichkeit und verwandte Gebiete}}, vol.~25, pp. 239--252, 1973.

\bibitem{CsiszarKorner}
I.~Csisz\'{a}r and J.~K\"{o}rner, \emph{Information Theory: Coding Theorems for
  Discrete Memoryless Systems}.\hskip 1em plus 0.5em minus 0.4em\relax
  Budapest: Akad\'{e}mi Kiad\'{o}, 1982.

\bibitem{CsiszarN:88positivity}
I.~Csisz\'{a}r and P.~Narayan, ``The capacity of the arbitrarily varying
  channel revisited : Positivity, constraints,'' \emph{IEEE Transactions on
  Information Theory}, vol.~34, no.~2, pp. 181--193, 1988.

\bibitem{SarwateG:10csi}
\BIBentryALTinterwordspacing
A.~D. Sarwate and M.~Gastpar, ``Rateless codes for {AVC} models,'' \emph{IEEE
  Transactions on Information Theory}, vol.~56, no.~7, pp. 3105--3114, July
  2010. [Online]. Available: \url{http://dx.doi.org/10.1109/TIT.2010.2048497}
\BIBentrySTDinterwordspacing

\bibitem{Langberg:04focs}
M.~Langberg, ``Private codes or succinct random codes that are (almost)
  perfect,'' in \emph{Proceedings of the 45th Annual IEEE Symposium on
  Foundations of Computer Science (FOCS 2004)}, Rome, Italy, 2004.

\bibitem{Smith:07scrambling}
A.~Smith, ``Scrambling adversarial errors using few random bits, optimal
  information reconciliation, and better private codes,'' in \emph{Proceedings
  of the 2007 ACM-SIAM Symposium on Discrete Algorithms (SODA 2007)}, 2007.

\bibitem{Langberg:08}
M.~Langberg, ``Oblivious channels and their capacity,'' \emph{IEEE Transactions
  on Information Theory}, vol.~54, no.~1, pp. 424--429, 2008.

\bibitem{DJLS10}
M.~Langberg, S.~Jaggi, and B.~Dey, ``Coding against delayed adversaries,'' in
  \emph{Proceedings of the 2010 International Symposium on Information Theory
  (ISIT)}, 2010.

\bibitem{GuruswamiS:09stoc}
V.~Guruswami and A.~Smith, ``\sj{Codes for Computationally Simple Channels:
  Explicit Constructions with Optimal Rate},'' in \emph{\sj{Proceedings of the
  Symposium on the Foundations of Computer Science (FOCS)}}, 2010, pp.
  \sj{723--732}.

\bibitem{CsiszarN:88constraints}
I.~Csisz\'{a}r and P.~Narayan, ``Arbitrarily varying channels with constrained
  inputs and states,'' \emph{IEEE Transactions on Information Theory}, vol.~34,
  no.~1, pp. 27--34, 1988.

\bibitem{DeyJL:09allerton}
\sj{B.K. Dey and S. Jaggi and M. Langberg}, ``\sj{Codes against Online
  Adversaries},'' in \emph{\sj{Proceedings of the 47th Annual Allerton
  Conference on Communication, Control, and Computing}}, Monticello, IL, USA,
  September-October \sj2009.

\bibitem{LangbergJD09}
M.~Langberg, S.~Jaggi, and B.~Dey, ``Binary causal-adversary channels,'' in
  \emph{Proceedings of the 2009 International Symposium on Information Theory
  (ISIT)}, 2009.

\bibitem{Plo:60}
\sj{M. Plotkin}, ``\sj{Binary codes with specified minimum distance},''
  \emph{\sj{IRE Transactions on Information Theory}}, vol. \sj{6}, no. \sj{4},
  pp. \sj{445--450}, \sj{1960}.

\bibitem{AloS:11}
N.~Alon and J.~H. Spencer, \emph{The probabilistic method}.\hskip 1em plus
  0.5em minus 0.4em\relax Wiley-Interscience, 2011, vol.~73.

\bibitem{CT06}
T.~M. Cover and J.~A. Thomas, \emph{Elements of information theory},
  2nd~ed.\hskip 1em plus 0.5em minus 0.4em\relax New York, NY, USA:
  Wiley-Interscience, 2006.

\bibitem{turan}
P.~Tur\'{a}n, ``On the theory of graphs,'' \emph{Colloqium Mathematicum},
  vol.~3, pp. 19--30, 1954.

\bibitem{Dun:90}
\sj{William Dunham}, \emph{\sj{Journey through Genius: The Great Theorems of
  Mathematics}}.\hskip 1em plus 0.5em minus 0.4em\relax \sj{New York, NY USA}:
  \sj{Wiley}, \sj{1990}.

\end{thebibliography}
\end{document}